\documentclass[11pt,letterpaper]{article}
\usepackage[margin=1in]{geometry}

\usepackage{ifpdf}
\ifpdf
    \usepackage[pdftex]{graphicx}
    \usepackage[update]{epstopdf}
\else
	\usepackage{graphicx}
\fi
\usepackage{wrapfig,bbm}
\usepackage{enumitem,color}
\usepackage{multirow}
\usepackage{amsthm}

\newtheorem{theorem}{Theorem}

\newtheorem{lemma}{Lemma}

\usepackage[mathscr]{euscript}
\usepackage{eufrak}
\usepackage{bbold}
\usepackage{array}
\usepackage[cmex10]{amsmath}
\usepackage{amssymb}
\usepackage{amstext}
\usepackage{cite,setspace}
\usepackage{algorithm2e}
\usepackage{url}
\newtheorem{definition}{\textbf{Definition}}

\DeclareGraphicsExtensions{.pdf,.png,.jpg,.eps}

\DeclareMathOperator{\lcm}{\textbf{lcm}}
\DeclareMathOperator{\gcd1}{\textbf{gcd}}

\usepackage{mathtools}
\DeclarePairedDelimiter{\ceil}{\lceil}{\rceil}

\newcommand{\Fb}{\mathsf{F}}

\newcommand{\Fc}{\mathcal{F}}
\newcommand{\Gc}{\mathcal{G}}

\newcommand{\Nc}{\mathcal{N}}
\newcommand{\Qc}{\mathcal{Q}}
\newcommand{\Jc}{\mathcal{J}}
\newcommand{\Sc}{\mathcal{S}}
\newcommand{\Tc}{\mathcal{T}}
\newcommand{\Tcc}{\tilde{\mathcal{T}}}

\newcommand{\Xc}{\mathcal{X}}

\newcommand{\MDS}{\mathrm{MDS}}

\begin{document}

\title{Capacity-Achieving Private Information Retrieval Codes from MDS-Coded Databases with Minimum Message Size}


\author{Ruida Zhou, Chao Tian, Hua Sun, and Tie Liu}

\maketitle

\begin{abstract}
We consider constructing capacity-achieving linear codes with minimum message size for private information retrieval (PIR) from $N$ non-colluding databases, where each message is coded using maximum distance separable (MDS) codes, such that it can be recovered from accessing the contents of any $T$ databases. It is shown that the minimum message size (sometimes also referred to as the sub-packetization factor) is significantly, in fact exponentially, lower than previously believed. More precisely, when $K>T/\gcd1(N,T)$ where $K$ is the total number of messages in the system and $\gcd1(\cdot,\cdot)$ means the greatest common divisor, we establish, by providing both novel code constructions and a matching converse, the minimum message size as $\lcm(N-T,T)$, where $\lcm(\cdot,\cdot)$ means the least common multiple. On the other hand, when $K$ is small, we show that it is in fact possible to design codes with a message size even smaller than $\lcm(N-T,T)$. 
\end{abstract}

\section{Introduction}

The problem of private information retrieval (PIR), since its introduction \cite{PIRfirstjournal}, has attracted significant attention from researchers in the fields of theoretical computer science, cryptography,  information theory, and coding theory. In the classical PIR model, a user wishes to retrieve one of the $K$ available messages, from $N$ non-colluding databases, each of which has a copy of these $K$ messages. User privacy needs to be preserved during message retrieval, which requires that the identity of the desired message not be revealed to any single database. 
To accomplish the task efficiently, good codes should be designed to download the least amount of data per-bit of desired message, the inverse of which is referred to as the capacity of the PIR system. This capacity problem in the classical setting was settled recently \cite{Sun_Jafar_PIR}. 

In practical systems, the databases may suffer from failures, and are also constrained on the storage space. Erasure codes can be used to improve both storage efficiency and failure resistance. This consideration motivated the investigation of PIR from MDS-coded databases \cite{shah2014one,banawan2018capacity,Tajeddine_Rouayheb,jingke2017subScienceChina}, with coding parameter $(N,T)$, i.e., the messages can be recovered by accessing any $T$ databases. The capacity of PIR from MDS-coded databases (MDS-PIR) was characterized \cite{banawan2018capacity} as
\begin{align}
    C = \left(1+ \frac{T}{N}+ \cdots + \left( \frac{T}{N} \right)^{K-1} \right)^{-1}.\label{eqn:capacityformula}
\end{align} 

In a given code, the smallest required number of symbols in each message is called the message size $L$ (sometimes also referred to as the sub-packetization factor), which is an important factor impacting the practicality and efficiency of the code. A large message size implies that the message (or the data file in practice systems) needs to be large for such code to be applicable, which significantly restricts the possible usage scenarios. Moreover, a large message size also usually implies that the encoding and the decoding functions are more complex, which not only requires more engineering efforts to implement but also hinders the efficiency of the system operation. From a theoretical point of view, a code with a smaller message size usually implies a more transparent  coding structure, which can be valuable for related problems; see, e.g., \cite{Lin19weakly} for such an example. Thus codes with a smaller message size are highly desirable in both theory and practice.

The capacity-achieving code given in \cite{banawan2018capacity} requires $L=TN^K$, which can be extremely large for a system with even a moderate number of messages. The problem of reducing the message size of capacity-achieving codes was recently considered by Xu and Zhang \cite{jingke2017subScienceChina}, and it was shown that under the assumption that all answers are of the same length, the message size must satisfy $L\geq T (N/\gcd1(N,T))^{K-1}$. These existing results may have left the impression that capacity-achieving codes would necessitate a message size exponential in the number of messages. 

In this work, we show that the minimum message size for capacity-achieving PIR codes can in fact be significantly smaller than previously believed, by providing capacity-achieving linear codes with message size $L=\lcm(N-T,T)$. 
Two linear code constructions, referred to as Construction-A and Construction-B, respectively, are given. The two constructions have the same download cost and message size, however Construction-B has a better upload cost (i.e., a lower communication cost for the user to send the queries), at the expense of being slightly more sophisticated than Construction-A. The key difference between the two proposed constructions and existing codes in the literature is that the proposed codes reduce the reliance on the so-called variety symmetry \cite{tian2018capacity}, which should be distinguished from the asymmetry discussed in \cite{banawan2019asymmetry}, and the answers may be of different lengths\footnote{The download cost is measured in this work as the expected number of downloaded symbols (over all random queries), which is in line with the prevailing approach in the literature when PIR capacity is concerned \cite{Sun_Jafar_PIR,banawan2018capacity}, where the download cost is viewed as being equivalent to certain entropy term. However, if we instead measure the download cost by the maximum number of downloaded symbols (among all possible queries), which was the alternative and more stringent approach used in \cite{sun2017optimal} and \cite{jingke2017subScienceChina}, then the optimal minimum message sizes will need to be much larger. In a sense, having the more stringent requirement that the maximum download cost needs to match the PIR capacity forces certain symmetrization to be built in the code, which necessitates a significant increase in the message size.}. We further show that this is in fact the minimum message size when $K>T/\gcd1(N,T)$, the proof of which requires a careful analysis of the converse proof of the information-theoretic MDS-PIR capacity. Finally, we show that, when $K$ is small, it is in fact possible to design codes with a message size even smaller than $\lcm(N-T,T)$. 

The code constructions and converse proof reflect a reverse engineering approach which further extends \cite{tian2018capacity,Tian:16Computer}. Particularly, in \cite{tian2018capacity}, a similar approach was used to tackle the canonical PIR setting with replicated databases and a capacity-achieving PIR code with the minimum message size and upload cost was discovered, and in the current work the databases are instead MDS-coded. The analysis technique and the code construction in the current work, however, are considerably more involved due to the additional coding requirements and the several integer constraints.

The rest of the paper is organized as follows. In Section~\ref{sec:sys}, a formal problem definition is given. Construction-A and Construction-B are then given in Section~\ref{sec:cod-A} and Section~\ref{sec:cod-B}, respectively, where the correctness and performance are also proved and analyzed. The optimality of message size is established by first identifying several critical properties of capacity-achieving codes in Section~\ref{sec:pro}, then lower-bounding the minimum message size when $K>T/\gcd1(N,T)$ in Section \ref{sec:lowerbound}. A special code is given in Section \ref{sec:small} to show that when $K\leq T/\gcd1(N,T)$, the message size can be even lower than $\lcm(N-T, T)$. Finally, Section \ref{sec:con} concludes the paper. Several technical proofs are relegated to the Appendices.

\section{System Model}
\label{sec:sys}

There are a total of $K$ mutually independent messages $W^0, W^1, \ldots, W^{K-1}$ in the system. Each message is uniformly distributed over $\Xc^L$, {\em i.e.,} the set of length-$L$ sequences in the finite alphabet $\Xc$. All the messages can be collected and written as a single length-$LK$ row vector $W^{0:K-1}$. Each message is MDS-coded and then distributed to $N$ databases, such that from any $T$ databases, the messages can be fully recovered. Since the messages are $(N,T)$ MDS-coded, it is without loss of generality to assume that $L=M\cdot T$ for some integer $M$.

When a user wishes to retrieve a particular message $W^{k^*}$, $N$ queries $Q_{0:N-1}^{[k^*]}=(Q_0^{[k^*]},\ldots,Q_{N-1}^{[k^*]})$ are sent to the databases, where $Q_n^{[k^*]}$ is the query for database-$n$. The retrieval needs to be information theoretically private, i.e., any database is not able to infer any knowledge as to which message is being requested. For this purpose, a random key $\Fb$ in the set $\Fc$ is used together with the desired message index $k^*$ to generate the set of queries $Q_{0:N-1}^{[k^*]}$. Each query $Q_n^{[k^*]}$ belongs to the set of allowed queries for database-$n$, denoted as $\Qc_n$. After receiving query $Q_{n}^{[k^*]}$, database-$n$ responds with an answer $A_n^{[k^*]}$. Each symbol in the answers also belongs to the finite field $\Xc$, and the answers may have multiple (and different numbers of) symbols. Using the answers $A^{[k^*]}_{0:N-1}$ from all $N$ databases, together with $\Fb$ and $k^*$, the user then reconstructs $\hat{W}^{k^*}$. 

A more rigorous definition of the linear information retrieval process we consider in this work can be specified by a set of coding matrices and functions as follows. For notational simplicity, we denote the cardinality of a set $\mathcal{A}$ as $|\mathcal{A}|$. 
\begin{definition}\label{def:fun}
A linear private information retrieval code from linearly MDS-coded databases (a linear MDS-PIR code) consists of the following coding components: 
\begin{enumerate}
\item A set of MDS encoding matrices:
\begin{align}
&\tilde{G}_n := \mathrm{diag}(\tilde{G}_n^0, \tilde{G}_n^1, \ldots, \tilde{G}_n^{K-1}), \notag \\
&\qquad\qquad\qquad\qquad  n \in \{0,1,\ldots,N-1\}, 
\end{align}
where $\tilde{G}_n^k $, $k \in \{0,1,\ldots,K-1\}$ is an $L \times M$ matrix in $\mathcal{X}$ for encoding message $W^k$, i.e., each message is not mixed with other messages during storage, and 
each $\tilde{G}_n$ encodes the messages into the information to be stored at database-$n$, denoted as $V_n=W^{0:K-1} \cdot \tilde{G}_n$;

\item A set of MDS decoding recovery functions: 
\begin{eqnarray}
\Psi_{\mathcal{T}}: \Xc^{LK} \rightarrow \Xc^{LK}, 
\end{eqnarray}
for each $\mathcal{T}\subseteq  \{0,1,\ldots,N-1\}$ such that $|\mathcal{T}|=T$, whose outputs are denoted as $\tilde{W}_{\mathcal{T}}^{0:K-1}=\Psi_{\mathcal{T}}(\{V_n:n\in \mathcal{T}\})$;
\item{A query function
\begin{eqnarray*}
&\phi_n: \{0,1,\ldots,K-1\} \times \Fc \rightarrow \Qc_n, \notag \\ 
&\qquad\qquad\qquad\qquad  n \in \{0,1,\ldots,N-1\},
\end{eqnarray*}
i.e., for retrieving message $W^{k^*}$,  the user sends the query $Q_{n}^{[k^*]} = \phi_n(k^*, \Fb)$ to database-$n$;}
\item{An answer length function
\begin{eqnarray}
\ell_n: \Qc_n \rightarrow \{0, 1, \ldots \},~\quad n \in \{0,1,\ldots,N-1\},
\end{eqnarray}
i.e., the length of the answer from each database, a non-negative integer, is a deterministic function
of the query, but not the particular realization of the messages;}
\item{
An answer generating matrix
\begin{align}
\hat{G}_n^{(q_n)} \in \Xc^{MK \times \ell_n}, \quad  q_n \in \Qc_n, n \in \{0,1,\ldots,N-1\},
\end{align}
i.e., the answer $A_n^{[k^*]}=A_n^{(q_n)} := V_n \cdot \hat{G}_n^{(q_n)}$, when $q_n=Q_n^{[k^*]}$ is the  query received by   database-$n$;}
\item{
A reconstruction function
\begin{eqnarray}
\psi: \prod_{n=0}^{N-1} \Xc^{\ell_n} \times \{0,1,\ldots,K-1\} \times \Fc \rightarrow \Xc^{L},
\end{eqnarray}
i.e., after receiving the answers, the user reconstructs the message as $\hat{W}^{k^*} = \psi(A_{0:N-1}^{[k^*]}, k^*, \Fb)$.
}
\end{enumerate}
These functions satisfy the following three requirements:
\begin{enumerate}
    \item {\bf MDS recoverable:} For any $\mathcal{T}\subseteq \{0,1,\ldots,N-1\} $ such that $|\mathcal{T}|=T$, we have $ \tilde{W}_{\mathcal{T}}^{0:K-1}=W^{0:K-1}$. 
    \item {\bf Retrieval correctness:} For any $k^* \in \{0,1,\ldots,K-1\} $, we have $\hat{W}^{k^*} = W^{k^*}.$
    \item {\bf Privacy:} For every $k, k' \in\{0,1,\ldots,K-1\} $, $n \in\{0,1,\ldots,N-1\} $ and $q \in \Qc_n$, 
    \begin{eqnarray}
    \mathbf{Pr}(Q_n^{[k]} = q) = \mathbf{Pr}(Q_n^{[k']} = q).
    \end{eqnarray}
\end{enumerate}
\end{definition}

Note that 
$Q_{n}^{[k^*]}$ is in fact a random variable, since $\Fb$ is the random key. It follows that even when the messages are viewed as deterministic, $A^{[k^*]}_{n} $ is still not deterministic. In contrast, for any specific query realization $Q_{n}^{[k^*]}=q_n$, the corresponding answer $A^{(q_n)}_n$ is deterministic when the messages are viewed as  deterministic. The distinction between $A^{[k^*]}_{n}$ and $A^{(q_n)}_n$ is indicated by the  bracket ${[\cdot]}$ and the parenthesis ${(\cdot)}$. 

In order to measure the performance of an MDS-PIR code, we consider the following two metrics, with the focus on minimizing the latter while keeping the former optimal: 
\begin{enumerate}
\item The retrieval rate, which is defined as 
\begin{eqnarray}
R:=\frac{L}{ \sum_{n=0}^{N-1} \mathbb{E} (\ell_n) }.\end{eqnarray}
This is the number of bits of desired message information that can be privately retrieved per bit of downloaded data. It was shown \cite{banawan2018capacity} that the maximum retrieval rate, i.e., the capacity of such MDS-PIR systems, is as given in (\ref{eqn:capacityformula}).
\item The message size $L$, which is the number of symbols to represent each individual message. This quantity should be minimized, because in practical applications, a smaller message size implies a more versatile code.
\end{enumerate}

A third metric, the upload cost, is also of interest in practical systems (also particularly in computer science literature, e.g., \cite{PIRfirstjournal}), although it is not our main focus in this work. The upload cost can be defined as 
\begin{align}
\sum_{n=0}^{N-1} \log_2  |\mathcal{Q}_n|,
\end{align}
which is roughly the total number of bits that the user needs to send to the servers during the query phase.

We will need several more parameters before proceeding. Define $p:=\gcd1(N,T)$, then 
\begin{eqnarray}
N-T = p\cdot r,\qquad \qquad T = p\cdot s, \label{eqn:prs}
\end{eqnarray}
for some positive integers $r$ and $s$, which are co-prime. 

\section{New MDS-PIR Code: Construction-A} \label{sec:cod-A}

In this section, we provide the first MDS-PIR code construction with message length $L=\lcm (N-T,T)$, which we refer to as Construction-A. 

\subsection{The Coding Components of Construction-A}

Each message $W^k$ can be divided into $M$ sub-messages, denoted as $W^{k}=(W^{k,0},W^{k,1},\ldots,W^{k,M-1})$, and each sub-message contains $T$ symbols in the alphabet $\Xc$. The construction relies on two novel ingredients: a new indexing on the key (query)  and the introduction of pseudo code symbols. The two elements were not present in other constructions in the literature such as \cite{Sun_Jafar_PIR,banawan2018capacity,jingke2017subScienceChina}. A simpler version of these two ingredients were first used in \cite{tian2018capacity} for replicated databases. The generalized version used in the current work requires a more complex translation between indexing and the answer, as well as the introduction of more than one pseudo code symbol. 

The \textit{first novel ingredient} in the construction, which is different from previous ones in the literature, is the random key $\Fb = (\Fb_0,\Fb_1,\ldots, \Fb_{K-1})$, which is a length-$K$ vector uniformly distributed in the set 
\begin{align}
&\Fc:= \bigg{\{}(f_0,\ldots,f_{K-1}) \in \{0,\ldots,r+s-1\}^K \notag \\
& \qquad \qquad \qquad\qquad \qquad \qquad \bigg{|} \left(\sum_{k=0}^{K-1}f_k\right)_{r+s}= 0  \bigg{\}}, \label{eqn:Fc}
\end{align} 
where $(\cdot)_{r+s}$ indicates modulo $(r+s)$. In this code construction, we need to first choose a (in fact, any) linear $(N,T)$-MDS code $\mathbb{C}$, in the alphabet $\Xc$ as our base code. There are many known techniques to construct such codes, such as Reed-Solomon codes and Cauchy matrix based constructions; see \cite{LinCostello:book}.  The coding functions can then be given as follows:

\begin{figure}[bt]
\centering
\includegraphics[width=0.75\textwidth]{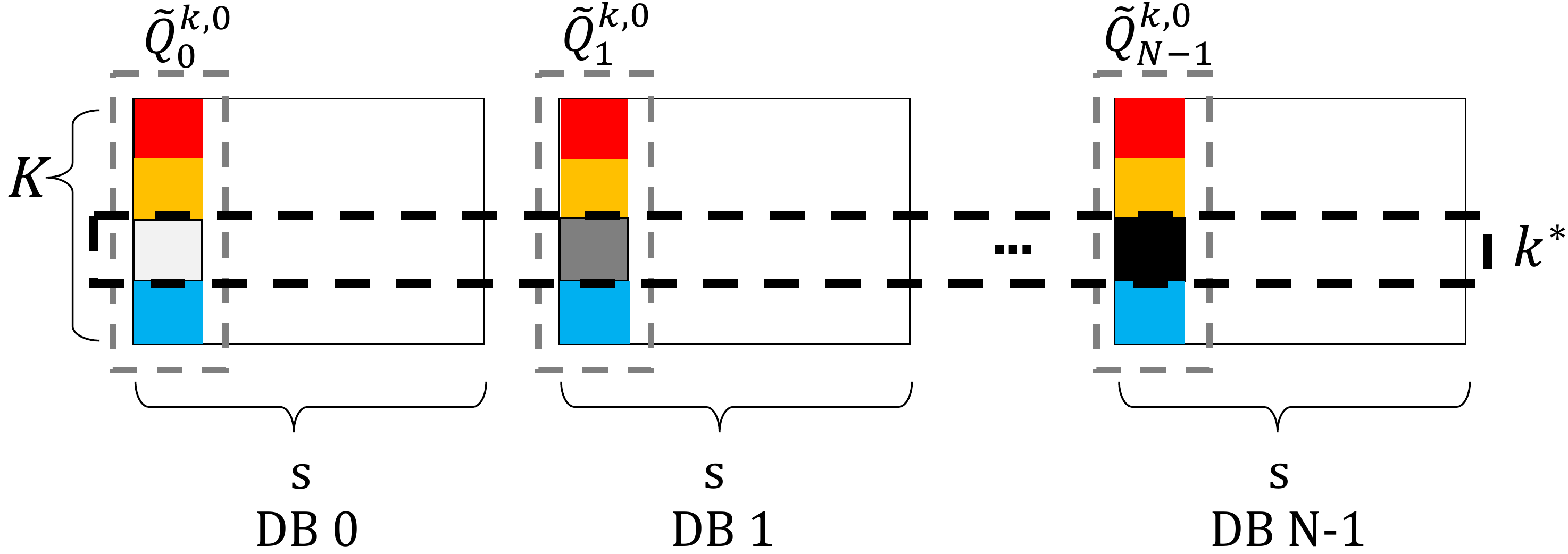}
\caption{The queries to different databases are illustrated. The parts of the queries related to the interference signals are of the same pattern. As a consequence, the induced interference signals in the answers will have the same pattern, and $T$ of them can be isolated to remove the interference signals in all the answers. \label{fig:interferencecancel}}
\end{figure}

\begin{enumerate}
    \item Each sub-message $W^{k,m}$, $m=0,1,\ldots, r-1$ and $k=0,1,\ldots,K-1$, is encoded by $\mathbb{C}$ into $N$ coded symbols $V^{k,m}_{0:N-1}=(V^{k,m}_{0},V^{k,m}_{1},\ldots,V^{k,m}_{N-1})$, with $V^{k,m}_{n}= W^{k,m}\cdot\tilde{G}^*_n \in \Xc$ placed at database-$n$, where $\tilde{G}^*_n$ is the $n$-th column of the $T\times N$ generator matrix of code $\mathbb{C}$ operated on each sub-message, which produces the stored information at database-$n$.  
  \item The $\MDS$ decoding function is obvious which is naturally induced by that of $\mathbb{C}$.
    \item For any $n \in \{0,1,\ldots, N-1\}$, the query generating function produces a length-$K$ column vector as
        \begin{align}
    &\phi_{n}(k^*,\Fb) =Q_{n}^{[k^*]}=\nonumber\\
     & (\Fb_0,\ldots,
    \Fb_{k^*-1},\left(\Fb_{k^*}+n\right)_{r+s},\Fb_{k^*+1},\ldots,\Fb_{K-1})^T.\label{eqn:queries}
    \end{align}
 
    \item Database-$n$ first produces a $K\times s$ query matrix $\tilde{Q}_{n}$ 
    \begin{align} 
    \tilde{Q}_{n}=\left(Q_{n}^{[k^*]}\cdot \mathbf{1}^T_s+\mathbf{1}_K\cdot [0,1,\ldots,s-1]\right)_{r+s},
    \end{align}
    where $\mathbf{1}_t$ is the all-one column vector of length $t$, and $^T$ indicates matrix transpose; the element of $\tilde{Q}_{n}$ on the $k$-th row and $i$-th column is denoted as $\tilde{Q}_{n}^{k,i}$. The query length function is then defined as:  
    \begin{eqnarray}
    \ell_n = \sum_{i=0}^{s-1} \mathbb{1} \left(\min_{k=0,1,\ldots,K-1} \tilde{Q}_n^{k,i}<r\right),\label{eqn:nonzeropos}
    \end{eqnarray}
    where $\mathbbm{1}(\cdot)$ is the indicator function, i.e., $\ell_n$ is the number of columns in $\tilde{Q}_{n}$ which has an element less than $r$. 
    
    \item  \textit{The second novel ingredient}, which is different from previous ones in the literature, is the introduction of pseudo code symbols and pseudo message symbols in the sub-messages: $V_n^{k,i}=W^{k,i}=0$ for $i\geq r$. For $n \in \{0,1,\ldots,N-1\}$, an intermediate answer vector $\tilde{A}_{n}^{[k^*]}$ of length-$s$ is formed as
    \begin{align}
    &\tilde{A}_{n}^{[k^*]} :=\left(\bigoplus_{k=0}^{K-1}V^{k,\tilde{Q}_{n}^{k,0}}_n, \bigoplus_{k=0}^{K-1}V^{k,\tilde{Q}_{n}^{k,1}}_n,\right.\nonumber\\
    &\left.\qquad\qquad\qquad\qquad\qquad\ldots,\bigoplus_{k=0}^{K-1}V^{k,\tilde{Q}_{n}^{k,s-1}}_n\right),
    \end{align}
each component of which is the finite field addition of some components of the vector $V_n$ that are indicated by the corresponding column of $\tilde{Q}_{n}$. The eventual answer ${A}_{n}^{[k^*]}$ of length $\ell_n$ is formed by concatenating the components of  $\tilde{A}_{n}^{[k^*]}$ which are not constantly zero, i.e., those corresponding to the positions indicated in (\ref{eqn:nonzeropos}).

    \item For any $i\in \{0,1,\ldots,s-1\}$, define the interference database set $\Tc_i := \{n|~\tilde{Q}_{n}^{k^*,i}\geq r \}$. The $i$-th component of $\tilde{A}_{n}^{[k^*]}$, $n\in \mathcal{T}_i$, can be written as 
    \begin{align*}
&    \bigoplus_{k=0}^{K-1} V_n^{k,\tilde{Q}_{n}^{k,i}}=\bigoplus_{k=0}^{K-1} \left(W^{k,\tilde{Q}_{n}^{k,i}}\cdot\tilde{G}^*_n\right) \nonumber\\
    &= \left(\bigoplus_{k=0}^{K-1}W^{k,\tilde{Q}_{n}^{k,i}}\right)\cdot\tilde{G}^*_n=\bar{W}^{[k^*],i}\cdot\tilde{G}^*_n, ~ n\in\Tc_i,
    \end{align*}     
    where the length-$T$ row vector  $\bar{W}^{[k^*],i}$ is defined as
    \begin{align*}
   \bar{W}^{[k^*],i} := \left(\bigoplus_{k=0}^{k^*-1} W^{k,\tilde{Q}_{n}^{k,i}}\oplus \bigoplus_{k=k^*+1}^{K-1} W^{k,\tilde{Q}_{n}^{k,i}}\right).
    \end{align*}
    Note that $\bar{W}^{[k^*],i}$ is not a function of $n$, since $\tilde{Q}_{n}^{k,i}=\tilde{Q}_{n'}^{k,i}$ unless $k=k^*$. Thus as long as $|\Tc_i|\geq T$, the vector $\bar{W}^{[k^*],i}$ can be fully recovered by the MDS property of the code $\mathbb{C}$; see Fig. \ref{fig:interferencecancel} for an illustration. Further note that the $i$-th component of $\tilde{A}_{n}^{[k^*]}$ for $n\in\{0,1,\ldots,N-1\}\setminus\Tc_i$ can be written as 
    \begin{align}
    \left(\bar{W}^{[k^*],i}\cdot\tilde{G}^*_n\right)\oplus \left(W^{k^*,\tilde{Q}_n^{k^*,i}}\cdot\tilde{G}^*_n\right),
    \end{align}
    from which, since $\bar{W}^{[k^*],i}$ is known, we can recover 
    \begin{align}
    \left(W^{k^*,\tilde{Q}_n^{k^*,i}}\cdot\tilde{G}^*_n\right), \quad n\in\{0,1,\ldots,N-1\}\setminus\Tc_i.
    \end{align}
  Denote $\Nc_m:= \left\{n\big{|} \tilde{Q}_{n}^{k^*,i}=m\right\}$. As long as $|\Nc_m|\geq T$, we can recover the vector $W^{k^*,m}$ by again invoking the property of the MDS code $\mathbb{C}$.  
\end{enumerate}

\begin{table*}[tb]
\def\arraystretch{1.3}
\setlength{\tabcolsep}{1.2pt}
 \centering
        \caption{Queries and answers for $(N,T,K)=(3,2,3)$.    \label{tbl:exp323}}

     \begin{tabular}{|c|c|c||c|c|c||c|c|c|}
        \hline
        \multicolumn{3}{|c||}{database-0}& \multicolumn{3}{c||}{database-1}&
        \multicolumn{3}{c|}{database-2} \\ \hline        
$Q_0$ & $\tilde{Q}_0$ &$A_0$ & $Q_1$& $\tilde{{Q}}_1$ &$A_1$ &${Q}_2$& $\tilde{{Q}}_2$ &$A_2$ \\ \hline
         
          $\begin{pmatrix} 
          0\\
          0\\
          0
         \end{pmatrix}$
         &
         $\begin{pmatrix} 
          01  \\
          01  \\
          01
         \end{pmatrix}$
         &
         $\begin{pmatrix} 
          V^0_0 \oplus V^1_0 \oplus V^2_0, & \emptyset
         \end{pmatrix}$
         &
         $\begin{pmatrix} 
          0  \\
          0 \\
          1
         \end{pmatrix}$
         &
         $\begin{pmatrix} 
          01 \\
          01\\
          12\\
         \end{pmatrix}$
         &
         $\begin{pmatrix} 
          V^0_1 \oplus V^1_1,& \emptyset
         \end{pmatrix}$         
          &
         $\begin{pmatrix} 
          \mathbf{0}\\
          \mathbf{0}\\
          \mathbf{2}
         \end{pmatrix}$
         &
         $\begin{pmatrix} 
          \mathbf{01}\\
          \mathbf{01}\\
          \mathbf{20}
         \end{pmatrix}$
         &
         $\begin{pmatrix} 
          \mathbf{V^0_2 \oplus V^1_2,} &  \mathbf{V^2_2}
         \end{pmatrix}$
         \\ \hline
         $\begin{pmatrix} 
          \mathbf{0}\\
          \mathbf{1}\\
          \mathbf{2}
         \end{pmatrix}$
         &
         $\begin{pmatrix} 
          \mathbf{01}\\
          \mathbf{12}\\
          \mathbf{20}
         \end{pmatrix}$
         &
         $\begin{pmatrix} 
          \mathbf{V^0_0,} &  \mathbf{V^2_0}
         \end{pmatrix}$
         &
         $\begin{pmatrix} 
          0\\
          1\\
          0
         \end{pmatrix}$
         &
         $\begin{pmatrix} 
          01\\
          12\\
          01
         \end{pmatrix}$
         &
         $\begin{pmatrix} 
          V^0_1 \oplus V^2_1,&  \emptyset
         \end{pmatrix}$
         &
         $\begin{pmatrix} 
          0\\
          1\\
          1
         \end{pmatrix}$
         &
         $\begin{pmatrix} 
          01\\
          12\\
          12
         \end{pmatrix}$
         &
         $\begin{pmatrix} 
          V^0_2,&\emptyset
         \end{pmatrix}$
         \\ \hline
         
          $\begin{pmatrix} 
          0\\
          2\\
          1
         \end{pmatrix}$
         &
         $\begin{pmatrix} 
          01\\
          20\\
          12
         \end{pmatrix}$
         &
         $\begin{pmatrix} 
          V^0_0,&  V^1_0
         \end{pmatrix}$
          &
         $\begin{pmatrix} 
          \mathbf{0}\\
          \mathbf{2}\\
          \mathbf{2}
         \end{pmatrix}$
         &
         $\begin{pmatrix} 
          \mathbf{01}\\
          \mathbf{20}\\
          \mathbf{20}
         \end{pmatrix}$
         &
         $\begin{pmatrix} 
          \mathbf{V^0_1,}&          \mathbf{V^1_1 \oplus V^2_1}
         \end{pmatrix}$
          &
         $\begin{pmatrix} 
          0\\
          2\\
          0 
         \end{pmatrix}$
         &
         $\begin{pmatrix} 
          01\\
          20\\
          01
         \end{pmatrix}$
         &
         $\begin{pmatrix} 
          V^0_2 \oplus V^2_2,&          V^1_2
         \end{pmatrix}$
         \\ \hline

\vdots&  \vdots                              & \vdots                          & \vdots & \vdots& \vdots                              & \vdots                          & \vdots & \vdots \\   \hline   
         $\begin{pmatrix} 
          2\\
          2\\
          2
         \end{pmatrix}$
         & 
         $\begin{pmatrix} 
          20\\
          20\\          
          20
         \end{pmatrix}$
         &
         $\begin{pmatrix} 
          \emptyset,&       V^0_0 \oplus V^1_0 \oplus V^2_0
         \end{pmatrix}$
          &
         $\begin{pmatrix} 
          2\\
          2\\
          0
         \end{pmatrix}$
         &
         $\begin{pmatrix} 
          20\\
          20  \\
          01
         \end{pmatrix}$
         &
         $\begin{pmatrix} 
          V^2_1,&  V^0_1 \oplus V^1_1
         \end{pmatrix}$
          &
         $\begin{pmatrix} 
          2\\
          2\\
          1
         \end{pmatrix}$
         &
         $\begin{pmatrix} 
          20\\
          20\\
          12
         \end{pmatrix}$
         &
         $\begin{pmatrix} 
          \emptyset,&      V^0_2 \oplus V^1_2
         \end{pmatrix}$
         \\ \hline
    \end{tabular}
\end{table*}

\subsection{An Example for Construction-A}

Let us first consider an example $(N,T,K)=(3,2,3)$, which induces $(p,r,s,L)=(1,1,2,2)$ in the code. We omit the index $i$ since here $r=1$. The possible queries $Q_0$, $Q_1$, and $Q_2$ are listed in the corresponding columns in Table \ref{tbl:exp323}. With a given query $Q_n$, the expanded query $\tilde{Q}_n$ is given to its right, the second column of which is  by adding $1$ to each component and then taking modulo-$3$, as specified in Step-4 of the protocol. The answer $A_n$ is then simply constructed by taking each column of $\tilde{Q}_n$, and forming the addition of the corresponding $V$ symbols, by however, taking advantage of the fact that $V_n^k=0$ whenever $k\geq 1$.

Consider the case to retrieve message $k^*=1$, and the key is $\Fb=(0,1,2)^T$. Then the 
queries are 
\begin{align}
Q_0=(0,1,2)^T,\quad Q_1=(0,2,2)^T,\quad Q_2=(0,0,2)^T.
\end{align}
The corresponding queries (and query matrices such induced) and answers are marked bold in Table~\ref{tbl:exp323}. In these $\tilde{Q}$ matrices, each column has at least one element being $0$, and thus the total number of transmission symbols is $6$.  It is seen that from $(V_0^0,V_1^0)$,  the symbol $V_2^0$ can be recovered by the MDS property, and thus $V_2^1$. Similarly, we can recover $V_1^1$. Using both $(V_1^1,V_2^1)$, we can then recover the original message $W_1$ by decoding the MDS code $\mathbb{C}$.

\subsection{Correctness, Privacy, and Download Cost}

According to the last coding component function (the reconstruction function) in Construction-A, the correctness of the proposed code hinges on two conditions: $|\mathcal{T}_i|\geq T$ for all $i=0,1,\ldots,s-1$ and $|\Nc_m|\geq T$ for all $m=0,1,\ldots,r-1$. We establish these two conditions in the following lemma, whose proof can be found in Appendix \ref{appendix:forward}.

\begin{lemma}\label{lemma:TN} 
   In Construction-A, for any request of message-$k^*$ and any random key $\Fb$, 
    \begin{enumerate}
        \item $|\Tc_i| = T$ for any $i\in \{0,1,\ldots,s-1\}$;
        \item $|\Nc_m| = T$ for any $m \in \{0,1,\ldots, r-1\}$.
    \end{enumerate}
\end{lemma}
 
We have the following main theorem for Construction-A. 

\begin{theorem} \label{thm:1}
Codes obtained by Construction-A are both private and capacity-achieving.
\end{theorem}
\begin{proof} 
The fact that the code is private is immediate, by observing that $Q_{n}^{[k^*]}$ is uniformly distributed on the set 
\begin{align}
&\mathcal{Q}_n=\bigg{\{}(f_0,\ldots,f_{K-1})^T \in \{0,\ldots,r+s-1\}^K \nonumber \\
& \qquad\qquad\qquad \qquad \qquad \bigg{|} \left(\sum_{k=0}^{K-1}f_k-n\right)_{r+s}=0  \bigg{\}}, \label{eqn:F}
\end{align}
regardless of the value of $k^*$.

The expected lengths of the answers is
\begin{align}
&\sum_{n=0}^{N-1} \mathbb{E} (\ell_n)=\sum_{n=0}^{N-1}\sum_{i=0}^{s-1} \mathbf{Pr}\left(\min_{k=0,1,\ldots,K-1} \tilde{Q}_n^{k,i}<r\right)\nonumber\\
&\qquad\qquad=\sum_{i=0}^{s-1}\sum_{n=0}^{N-1}\mathbf{Pr}\left(\min_{k=0,1,\ldots,K-1} \tilde{Q}_n^{k,i}<r\right),
\end{align}
assuming an arbitrary message $k^*$ is being requested. The probabilities involved in the summand $i=i^*$ depend on
\begin{align}
\left(\tilde{Q}_0^{0:K-1,i^*},\tilde{Q}_1^{0:K-1,i^*},\ldots,\tilde{Q}_{N-1}^{0:K-1,i^*}\right). \label{eqn:subqueries}
\end{align}
By the definition of $\tilde{Q}_n^{k,i}$, it is clear that if any item in
\begin{align*}
(\Fb_0+i^*,...,\Fb_{k^*-1}+i^*,\Fb_{k^*+1}+i^*,\ldots,\Fb_{K-1}+i^*)_{r+s}
\end{align*}
is less than $r$, then $\min_{k=0,1,\ldots,K-1} \tilde{Q}_n^{k,i^*}<r$ for all $n=0,1,\ldots,N-1$, which will induce $N$ transmitted symbols in the retrieval from all databases for $i=i^*$; this event $E$ occurs with probability $1-(s/(r+s))^{K-1}$. 
On the other hand, when the event $E$ does not occur, in the vector
\begin{align*}
(\Fb_{k^*}+i^*+0,\Fb_{k^*}+i^*+1,\ldots,\Fb_{k^*}+i^*+N-1)_{r+s}
\end{align*}
the number of elements that are less than $r$ is $N-T$, which induces $(N-T)$ symbols being transmitted. Therefore
\begin{align}
&\sum_{n=0}^{N-1} \mathbb{E} (\ell_n)=s\left[\mathbf{Pr}(E)N+\left(1-\mathbf{Pr}(E)\right)(N-T)\right]\nonumber\\
&=sN-sT\left(\frac{T}{N}\right)^{K-1}=sN\left[1-\left(\frac{T}{N}\right)^{K}\right],
\end{align}
from which it follows that the code is indeed capacity achieving, by taking into account (\ref{eqn:prs}).
\end{proof}

The following lemma is also immediate, and we state it as a lemma below. 
\begin{lemma}
The upload cost of Construction-A is $N(K-1)\log\left[N/\gcd1(N,T)\right]$. 
\end{lemma}
\begin{proof}
Consider any $k^*$. By (\ref{eqn:queries}), we see that $| \Qc_n | = |\Fc| = (N/\gcd1(N,T))^{K-1}$, and it follows that the upload cost is $\sum_{n=0}^{N-1} \log(|\Qc_n|) = N(K-1)\log\left[N/\gcd1(N,T)\right]$.
\end{proof}

\section{New MDS-PIR Code: Construction-B} \label{sec:cod-B}

In this section, we provide an alternative code construction, namely Construction-B. This construction requires a lower upload cost than Construction-A, however, it relies on two different coding strategies for the two cases of high rate codes $T\geq  N-T$ and low rate codes $T\leq N-T$. The high rate code construction is essentially built  on a product code, while the low rate codes bear more similarity to Construction-A.

\subsection{Construction-B for $T\geq N-T$}
In this construction, the same random key $\Fb = (\Fb_0,\Fb_1,\ldots, \Fb_{K-1})$ as in Construction-A is used, and the MDS encoding matrices and decoding functions are also exactly the same as in Construction-A. We need a second generic $(s,r)$-$\MDS$ code $\mathbb{C}_c$ in the alphabet $\mathcal{X}$ in this construction. Construction-B essentially utilizes a product code with row code $\mathbb{C}$ and column code $\mathbb{C}_c$ \cite{LinCostello:book}. In this context, it is more convenient to view the message $W^k$ as being represented as an $r\times T$ matrix, denoted as $\breve{W}^k$
\begin{align}
\breve{W}^k=\begin{bmatrix}
W^{k,0}\\
W^{k,1}\\
\vdots\\
W^{k,r-1}
\end{bmatrix}.
\end{align}
Next we provide the  coding components $(3-6)$ in Construction-B.

\begin{enumerate}
\setcounter{enumi}{2}
    
    \item The query generating function at server-$n$ produces the following $K\times 1$ query vector
    \begin{align}
    &\phi_{n}(k^*,\Fb) = Q^{[k^*]}_{n}=(Q^{[k^*]}_{0,n},Q^{[k^*]}_{1,n},\ldots,Q^{[k^*]}_{K-1,n})^T\nonumber\\
    &\qquad = \ceil{(\Fb_0,\Fb_1,\ldots,  \Fb_{k^*-1},\left(\Fb_{k^*}+n\right)_{s+r},\nonumber\\
    &\qquad\qquad\qquad\qquad\qquad\Fb_{k^*+1},\ldots,\Fb_{K-1})^T}_s,
    \end{align}    
    where $\ceil{x}_s :=\min(x,s)$, and it operates on a vector by operating on each component individually. 
    \item Define an $s \times (s+1)$ query pattern matrix $P$ as
    \begin{align}
    P:=\left[
    \begin{array}{ccccccccc|c}
    1&1&1&\cdots&1&0&0&\cdots&0&0\\
    0&1&1&\cdots&1&1&0&\cdots&0&0\\
    \vdots&\vdots&\vdots&\vdots&&\vdots&&\vdots&\vdots&0\\ 
    1&1&\cdots&1&0&0&\cdots&0&1&0
    \end{array}
        \right],
    \end{align}
    where the first row has the first $r$ elements being $1$'s and the rest $(s-r+1)$ being $0$'s, and the remaining rows are obtained by cyclically shifting the first $s$ elements in the first row but keeping the last $0$ in place. 
    The query length function is then defined as
    \begin{align}
        \ell_n  = \sum_{i=0}^{s-1} \mathbb{1} \left( \sum_{k=0}^{K-1} P_{i, Q^{[k^*]}_{k,n}} > 0 \right),\label{eqn:lengthB1}
     \end{align}
     i.e., it is the number of columns in the matrix $P$ selected by the vector $Q^{[k^*]}_{n}$ that have non-zero elements. 
    \item 
    Recall that coded message $W^k$ at database-$n$ is a length-$r$ vector $V^k_n$
    \begin{align}
   V^k_n= \begin{bmatrix}
    V^{k,0}_{n}\\  V^{k,1}_{n}\\ \vdots\\ V^{k,r-1}_{n}
        \end{bmatrix}=  \begin{bmatrix}
    W^{k,0} \\   W^{k,1}\\ \vdots\\  W^{k,r-1}
    \end{bmatrix}\cdot \tilde{G}^*_n.
    \end{align}
        In order to generate the answer, each $V_n^{k}$ vector is encoded by $\mathbb{C}_c$ into a length-$s$ intermediate code vector 
        \begin{align}
        \tilde{V}^k_{n}=(\tilde{V}^k_{n,0},\tilde{V}^k_{n,1},\ldots,\tilde{V}^k_{n,s-1}) =(\hat{G}^*)^T\cdot V^k_n,
        \end{align}
        where $\hat{G}^*$ is the generator matrix of the code $\mathbb{C}_c$.  An intermediate answer vector is then produced 
    \begin{align}
    &\tilde{A}^{[k^*]}_{n} := \left(\bigoplus_{k=0}^{K-1}  \tilde{V}^{k,0}_{n} \cdot P_{0, Q^{[k^*]}_{k,n}}, \bigoplus_{k=0}^{K-1}  \tilde{V}^{k,1}_{n} \cdot P_{1, Q^{[k^*]}_{k.n}},\right.\nonumber\\
    &\qquad\qquad\qquad\left.\cdots,\bigoplus_{k=0}^{K-1}  \tilde{V}^{k,s-1}_{n} \cdot P_{s-1, Q^{[k^*]}_{k,n}}\right)^T.
    \end{align}
    The eventual answer ${A}_{n}^{[k^*]}$ of length $\ell_n$ is formed by concatenating the components of $\tilde{A}^{[k^*]}_{n}$ 
    which are not constantly zero, i.e., those indicated by (\ref{eqn:lengthB1}).
    \item For any $i\in \{0,1,\ldots,s-1\}$, define the interference database set $\Tcc_i := \{n|~P_{i, Q_{k^*,n}^{[k^*]}} = 0\}$. For $n\in \Tcc_i$, the $i$-th symbol in the intermediate answer is
    \begin{align*}
    &\tilde{A}^{[k^*]}_{n,i} = \bigoplus_{k=0}^{K-1}  \tilde{V}^{k,i}_{n} \cdot P_{i, Q^{[k^*]}_{k,n}}\nonumber\\
    &\qquad= \bigoplus_{k=0}^{K-1}  (\hat{G}^*_i)^T \cdot {V}^{k}_{n}\cdot P_{i, Q^{[k^*]}_{k,n}}\nonumber\\
    &\qquad=\bigoplus_{k=0}^{K-1}  (\hat{G}^*_i)^T \cdot \breve{W}^{k}\cdot \tilde{G}^*_n\cdot P_{i, Q^{[k^*]}_{k,n}}\nonumber\\
    &\qquad=\bigoplus_{k=0}^{K-1}  (\hat{G}^*_i)^T \cdot \left(\breve{W}^{k}\cdot P_{i, Q^{[k^*]}_{k,n}}\right)\cdot \tilde{G}^*_n\nonumber\\
    &\qquad=   (\hat{G}^*_i)^T \cdot \left(\bigoplus_{k=0}^{K-1}\breve{W}^{k}\cdot P_{i, Q^{[k^*]}_{k,n}}\right)\cdot \tilde{G}^*_n\nonumber\\
    &\qquad=(\hat{G}^{*}_i)^T \cdot \bar{W}^{[k^*]}_i\cdot\tilde{G}^*_n,   
    \end{align*}     
    where the $r\times T$ matrix $ \bar{W}^{[k^*]}_i$ is defined as
    \begin{align*}
   &\bar{W}^{[k^*]}_i := \left(\bigoplus_{k=0}^{k^*-1}\breve{W}^{k}\cdot P_{i, Q^{[k^*]}_{k,n}} \right) \nonumber\\
   &\qquad\qquad\qquad\qquad\qquad\oplus \left(\bigoplus_{k=k^*+1}^{K-1} \breve{W}^{k}\cdot P_{i, Q^{[k^*]}_{k,n}}\right). 
    \end{align*}
    Note that $\bar{W}^{[k^*]}_i$ is not a function of $n$, since $Q_{k, n}^{[k^*]}=Q_{k, n'}^{[k^*]}$ unless $k=k^*$. Thus as long as $|\Tcc_i|\geq T$, the vector $(\hat{G}^{*}_i)^T \cdot \bar{W}^{[k^*]}_i$ can be fully recovered by the MDS property of the code $\mathbb{C}$. Further note that $\tilde{A}_{n,i}^{[k^*]}$ for $n\in\{0,1,\ldots,N-1\}\setminus\Tcc_i$ can be written as 
    \begin{align}
    \left((\hat{G}^{*}_{i})^T\cdot \bar{W}^{[k^*]}_i\cdot\tilde{G}^*_n\right)\oplus \left((\hat{G}^{*}_{i})^T\cdot \breve{W}^{k^*}\cdot\tilde{G}^*_n\right),
    \end{align}
    from which, since $(\hat{G}^{*}_i)^T \cdot \bar{W}^{[k^*]}_i$ is known, we can recover 
    \begin{align}
    (\hat{G}^{*}_i)^T \cdot \breve{W}^{k^*}\cdot\tilde{G}^*_n, \quad n\in\{0,1,\ldots,N-1\}\setminus\Tcc_i. \label{eqn:usefuln}
    \end{align}
    Denote $\mathcal{S}_n:= \left\{i \Big{|} P_{i, Q_{k^*,n}^{[k^*]}}=1 \right\}$  and $\mathcal{N}:= \left\{n \mid |\Sc_n|\geq r\right\}$,  the latter of which is the set of databases that provide at least $r$ symbols of the requested messages in the form (\ref{eqn:usefuln}). For any $n \in \Nc$, we can recover $\breve{W}^{k^*} \cdot \tilde{G}^*_n$ by invoking the property of $\MDS$ code $\mathbb{C}_c$. Then as long as $|\mathcal{N}| \geq T$,  $\breve{W}^{k^*}$ can be fully recovered by  invoking the $\MDS$ property of code $\mathbb{C}$. 
\end{enumerate}

\subsubsection{An Uncompressed Description of Construction-B}
The description of the coding components above is in a compressed form, and offers little intuition. The following equivalent description, on the other hand, can provide better intuition at the expense of more redundant items. Let the extended pattern matrix $\bar{P}$ of size $s \times (s+r)$ be defined as
\begin{align}
\bar{P}=[P|\mathbf{0}_{s\times(r-1)}],
\end{align}
i.e., expanding the original pattern matrix $P$ by appending an all-$0$ matrix of size $s\times(r-1)$. The same query answer can now be equivalently produced at each server by using the following auxiliary query
   \begin{align}
&\bar{Q}^{[k^*]}_{n} = (\Fb_0,\Fb_1,\ldots, \Fb_{k^*-1},\nonumber\\
&\qquad\qquad\qquad\left(\Fb_{k^*}+n\right)_{s+r},\Fb_{k^*+1},\ldots,\Fb_{K-1})^T,
    \end{align} 
    i.e., without using the $\ceil{\cdot}_s$ function mapping, and then following the same manner in answer generating, using the extended patter matrix $\tilde{P}$. The stored contents of message $W^k$ across all the databases can be visualized as follows
\begin{align}
\begin{bmatrix}
    \tilde{V}_{0}^{k,0}       & \tilde{V}_{1}^{k,0} & \cdots &  \tilde{V}_{N-1}^{k,0}\\
    \tilde{V}_{0}^{k,1}      & \tilde{V}_{1}^{k,1} & \cdots &  \tilde{V}_{N-1}^{k,1}\\
    \vdots &    \vdots&    \vdots&    \vdots\\
    \tilde{V}_{0}^{k,r-1}   & \tilde{V}_{1}^{k,r-1} & \cdots &  \tilde{V}_{N-1}^{k,r-1} \\\hline
    \tilde{V}_{0}^{k,r}      & \tilde{V}_{1}^{k,r} & \cdots &  \tilde{V}_{N-1}^{k,r} \\
    \vdots &    \vdots&    \vdots&    \vdots\\
    \tilde{V}_{0}^{k,s-1}   & \tilde{V}_{1}^{k,s-1} & \cdots &  \tilde{V}_{N-1}^{k,s-1}
\end{bmatrix},
\end{align}
where each column corresponds to a database, and the contents below the horizontal line are not stored physically, but can be generated as part of the answer computation. 

It is straightforward to verify that with the auxiliary query $\bar{Q}^{[k^*]}_{n}$ as the queries, the extended pattern matrix $\bar{P}$ in the answer generation, and the uncompressed stored contents $\tilde{V}_{0}^{k,0}$ as the stored content, the answer is precisely the same as the uncompressed version described above.

\subsubsection{An Example of Construction-B}
Consider an example $N=5$, $T=3$, $K=4$, which induces the parameters $(p,r,s,L)=(1,2,3,6)$. The pattern matrix $P$ and the extended pattern matrix $\bar{P}$ are
\begin{align}
P=\begin{bmatrix}
1&1&0&0\\
0&1&1&0\\
1&0&1&0
\end{bmatrix},\quad
\bar{P}=\begin{bmatrix}
1&1&0&0&0\\
0&1&1&0&0\\
1&0&1&0&0
\end{bmatrix}. 
\end{align}
Let the messages be $W^{0:3}=(A,B,C,D)$. Consider the case when message $W^0=A$ is being requested, and the key is $\Fb=(4,1,2)$. Then the auxiliary queries are
\begin{align}
[\bar{Q}_0,\bar{Q}_1,\bar{Q}_2,\bar{Q}_3,\bar{Q}_4]=\begin{bmatrix}
3&4&0&1&2\\
4&4&4&4&4\\
1&1&1&1&1\\
2&2&2&2&2
\end{bmatrix},
\label{eqn:Qbar}
\end{align}
and the compressed queries are 
\begin{align}
[{Q}_0,{Q}_1,Q_2,{Q}_3,Q_4]=\begin{bmatrix}
3&3&0&1&2\\
3&3&3&3&3\\
1&1&1&1&1\\
2&2&2&2&2
\end{bmatrix}.\label{eqn:Q}
\end{align}
The answers from the five databases are then 
\begin{align}
\left[
\begin{array}{ccccc}
C^0_{0}                &  C^0_{1}                              & A^0_{2}+C^0_{2}             & A^0_{3}+C^0_{3}                           & C^0_{4}\\
C^1_{0}+D^1_{0} & C^1_{1}+D^1_{1}              & C^1_{2}+D^1_{2}             & A^1_{3}+C^1_{3}+D^1_{3}            & A^1_{4}+C^1_{4}+D^1_{4}\\
D^2_{0}                & D^2_{1}                              & A^2_{2}+D^2_{2}             & D^2_{3}                                           & A^2_{4}+D^1_{4}
\end{array}
\right],
\end{align}
where we have used $A^m_n$ to denote $\tilde{V}^{0,m}_n$, as the corresponding coded message $W^0=A$, and similarly for $B,C,D$. Observer that in the first row $(C^0_{0},C^0_{1},C^0_{4})$ can be used to recover $C^0_n$, $n=0,1,\ldots,4$, and thus to obtain $(A^0_{2},A^0_{3})$; similarly, in the second row and third row, we can recover information on the $A$ message. The information on $A$ we can recover is thus as given in the following message matrix where each column corresponds to a database
\begin{align}
\left[
\begin{array}{ccccc}
*&*              & A^0_{2}         & A^0_{3}       & *\\
*&*              &*                     & A^1_{3}       & A^1_{4}\\
*&*              & A^2_{2}         &*                   & A^2_{4}
\end{array}
\right].
\end{align}
It is now straightforward to see that through the product code based on the $(5,3)$ MDS code $\mathbb{C}$ and the $(3,2)$ MDS code $\mathbb{C}_c$, the message $A$ can be fully recovered. 

\subsubsection{Correctness, Privacy, and Communication Costs}

Similar to Construction-A, the correctness of Construction-B relies on the following two facts established as Lemma \ref{lemma:Double}, whose proof can be found in Appendix \ref{appendix:forward}.
\begin{lemma}\label{lemma:Double} 
   In Construction-B, for any request of message-$k^*$ and any random key $\Fb$, 
    \begin{enumerate}
        \item $|\Tcc_i| = T$ for any $i\in \{0,1,\ldots,s-1\}$;
        \item $|\Nc| = T$.
    \end{enumerate}
\end{lemma}

We also have the follow main theorem for Construction-B. 
 
\begin{theorem} \label{thm:low1}
The codes obtained by Construction-B for $T\geq N-T$ are both private and capacity-achieving.
\end{theorem}
\begin{proof} 
To see that the code is private, observe that for database-$n$, the auxiliary query vector $\bar{Q}^{[k^*]}_{n}$
follows a uniform distribution in the set defined in (\ref{eqn:F}) for any requested message $k^* \in \{0,1,\ldots,K-1\}$. The query $Q^{[k^*]}_n$ is obtained through an additional mapping $ \ceil{\cdot}_s$ regardless of $k^*$, and thus the query at database-$n$ follows the same distribution for all $k^*$. 

The expected lengths of the answers can be written as
\begin{align}
&\sum_{n=0}^{N-1} \mathbb{E} (\ell_n)=\sum_{n=0}^{N-1}\sum_{i=0}^{s-1} \mathbf{Pr}\left( \sum_{k=0}^{K-1}P_{i,Q^{[k^*]}_{k,n}} > 0\right) \notag \\
&\qquad\,\,\qquad=\sum_{i=0}^{s-1}\sum_{n=0}^{N-1} \mathbf{Pr}\left( \sum_{k=0}^{K-1}P_{i,Q^{[k^*]}_{k,n}} > 0\right)\nonumber\\
&\qquad\,\,\qquad=\sum_{i=0}^{s-1}\sum_{n=0}^{N-1} \mathbf{Pr}\left( \sum_{k=0}^{K-1}\bar{P}_{i,\bar{Q}^{[k^*]}_{k,n}} > 0\right),
\end{align}
assuming an arbitrary message $k^*$ is being requested. The probabilities involved in the summand for $i=i^*$ depend on the vector
\begin{align}
\left(\bar{Q}_{0:K-1,0}^{[k^*]},\bar{Q}_{0:K-1,1}^{[k^*]},\ldots,\bar{Q}_{0:K-1,N-1}^{[k^*]}\right). \label{eqn:subqueries}
\end{align}
By the definition of $\bar{Q}_{0:K-1,n}^{[k^*]}$, it is clear that if  $\bar{P}_{i^*, \Fb_k}=1$ for any $k\in\{0,1,\ldots,k^*-1,k^*+1,\ldots,K-1\}$, 
then $\sum_{k=0}^{K-1}P_{i,Q^{[k^*]}_{k,n}} >0$ for all $n=0,1,\ldots,N-1$. This will induce $N$ transmissions in the retrieval from all databases for $i=i^*$. This event, denoted as $E$, occurs with probability $1-(s/(r+s))^{K-1}$, since the $i^*$-th row of the matrix $\bar{P}$ has $r$ entries of value $1$, and $\bar{Q}^{[k^*]}_{k,n}=\Fb_k$, $k=0,1,\ldots,k^*-1,k^*+1,\ldots,K-1$, are mutually independent, identically distributed, and each follows a uniform distribution on $\{0,1,\ldots,r+s-1\}$.

On the other hand, when the event $E$ does not occur, the vector
\begin{align*}
(\Fb_{k^*}+0,\Fb_{k^*}+1,\ldots,\Fb_{k^*}+N-1)_{r+s}
\end{align*}
is a permutation of the $p$-replicated vector of $(0,1,\ldots,s+r-1)$, and thus the number of elements that satisfying $\bar{P}_{i^*,(\Fb_{k^*}+n)_{r+s}}=1$, $n=0,1,\ldots,N-1$, is exactly $N-T$, implying that $(N-T)$ symbols  will be transmitted for $i=i^*$. Therefore, we have
\begin{align}
&\sum_{n=0}^{N-1} \mathbb{E} (\ell_n)=s\left[\mathbf{Pr}(E)N+\left(1-\mathbf{Pr}(E)\right)(N-T)\right]\nonumber\\
&\qquad=sN-sT\left(\frac{T}{N}\right)^{K-1}=sN\left[1-\left(\frac{T}{N}\right)^{K}\right],
\end{align}
from which it follows that the code is indeed capacity achieving.
\end{proof}

\begin{lemma}
\label{lemma:uploadcost}
   The upload cost of Construction-B for $T\geq N-T$ is upper-bounded by $$\min[N(K-1)\log(s+r),NK\log (s+1)].$$
\end{lemma}
\begin{proof}
For datebase-$n$, the query $Q_{n}^{[k^*]}$ has $K$ symbols, and each symbol is from the alphabet $\{0,1,\ldots,s\}$, and thus the the upload cost is clearly upper bounded by $N\log (s+1)^K$. However, observe that the query $Q_n$ is calculated by mapping the random key $\Fb$ with $\log|\Fc|= N(K-1)\log(s+r)$ through an surjective function $ \ceil{\cdot}_s$, whose image size is upper bounded by $\log (s+r)^{K-1}$. If $(s+1)^K>(s+r)^{K-1}$, then we can simply use the auxiliary query $\bar{Q}_{n}^{[k^*]}$. Thus the upload cost is at most the less of the two terms as given above. 
\end{proof}

If $r>1$, the quantity $NK\log (s+1)$ is clearly the less of the two when $K$ is large, and thus may lead to significant savings in terms of the upload cost. 

\subsection{Construction-B for $T\leq N-T$}

Here the same random key $\Fb = (\Fb_0,\Fb_1,\ldots, \Fb_{K-1})$ as in Construction-A is again used, and the MDS encoding matrices and decoding functions are also exactly the same as in Construction-A. The other components of the codes are as follows.
\begin{enumerate}
\setcounter{enumi}{2}
    \item 
    For any $n \in \{0,1,\ldots, N-1\}$, the query generating function produces a query with $K$ symbols
    \begin{align}
    &\phi_{n}(k^*,\Fb) ={Q}^{[k^*]}_{n} =(Q^{[k^*]}_{0,n}, Q^{[k^*]}_{1,n}, \ldots, Q^{[k^*]}_{K-1,n})^T =\nonumber\\
     & \ceil{(\Fb_0,\ldots,
    \Fb_{k^*-1},\left(\Fb_{k^*}+n\right)_{s+r},\Fb_{k^*+1},\ldots,\Fb_{K-1})^T}_r.
    \end{align}
    \item 
    The query length function is then defined as
    \begin{align}
        \ell_n  = s \cdot \mathbb{1} \left( \min_{k=0,\ldots,K-1} Q^{[k^*]}_{k,n} < r \right).\label{eqn:lengthB2}
    \end{align}
    \item Let $V_n^{k,r}=W^{k,r}=0$. Database-$n$ first produces a $K\times s$ query matrix $\tilde{Q}_{n}$ for  $i=0,1,\ldots,s-1$ 
    \begin{align} 
    \tilde{Q}_{n}^{k,i}=\left\{
    \begin{array}{ll}
    r & \text{if } Q^{[k^*]}_{k,n}=r \\
    \left(Q^{[k^*]}_{k,n}+i \right)_r &\text{otherwise}
    \end{array}
    \right..
    \end{align}
    For $n \in \{0,1,\ldots,N-1\}$, an intermediate answer vector $\tilde{A}_{n}^{[k^*]}$ of length-$s$ is formed (similar to Construction-A) as
    \begin{align}
    &\tilde{A}_{n}^{[k^*]} :=\left(\bigoplus_{k=0}^{K-1}V^{k,\tilde{Q}_{n}^{k,0}}_n, \bigoplus_{k=0}^{K-1}V^{k,\tilde{Q}_{n}^{k,1}}_n,\right.\nonumber\\
    &\qquad\qquad\qquad\qquad\qquad \left.\ldots,\bigoplus_{k=0}^{K-1}V^{k,\tilde{Q}_{n}^{k,s-1}}_n\right)^T.
    \end{align}
The eventual answer ${A}_{n}^{[k^*]}$ of length $\ell_n$ is formed by concatenating the components of  $\tilde{A}_{n}^{[k^*]}$ which are not constantly zero, as indicated by (\ref{eqn:lengthB2}).
    \item The reconstruction function is the same as that of Construction-A, and the desired message can be correctly reconstructed as long as $|\Tc_i|\geq T$ and $|\Nc_m|\geq T$.
\end{enumerate}

For better visualization, we can again consider the (uncompressed) auxiliary query 
\begin{align}
&\bar{Q}^{[k^*]}_{n} = (\Fb_0,\Fb_1,\ldots,
    \Fb_{k^*-1},\left(\Fb_{k^*}+n\right)_{s+r},\nonumber\\
    &\qquad\qquad\qquad\qquad\Fb_{k^*+1},\ldots,\Fb_{K-1})^T.
\end{align}
The query vector ${Q}^{[k^*]}_{n} $ is a compressed version of the auxiliary query $\bar{Q}^{[k^*]}_{n}$. 
 
\subsubsection{An Example for Construction-B}

Consider an example $N=5$, $T=2$, $K=4$, which induces the parameters $(p,r,s,L)=(1,3,2,6)$. 
Let the messages be $W^{0:3}=(A,B,C,D)$. Consider the case when message $W_0=A$ is being requested, and the key is $\Fb=(4,1,2)$. Then the auxiliary queries and the queries are as given in (\ref{eqn:Qbar}) and (\ref{eqn:Q}), respectively. The intermediate query matrix at all the databases are 
\begin{align}
&[\tilde{Q}_0,\tilde{Q}_1,\tilde{Q}_2,\tilde{Q}_3,\tilde{Q}_4]\nonumber\\
&\qquad=\left[
\begin{array}{cc|cc|cc|cc|cc}
3&3&3&3&0&1&1&2&2&0\\
3&3&3&3&3&3&3&3&3&3\\
1&2&1&2&1&2&1&2&1&2\\
2&0&2&0&2&0&2&0&2&0
\end{array}
\right].
\end{align}

The answers from the five databases are then 
\begin{align}
\left[
\begin{array}{c|c|c|c|c}
C^1_{0}+D^2_0    &  C^1_{1}+D^2_1         & A^0_{2}+C^1_{2}+D^2_2        & A^1_{3}+C^1_{3}+D^2_3              & A^2_{4}+C^1_{4}+D^2_4\\
C^2_{0}+D^0_{0} &  C^2_{1}+D^0_{1}      & A^1_{2}+C^2_{2}+D^0_{2}    & A^2_{3}+C^2_{3}+D^0_{3}           & A^0_{4}+C^2_{4}+D^0_{4}
\end{array}
\right].
\end{align}
In the first row, $(C^1_{0}+D^2_0, C^1_{1}+D^2_1)$ can be used to recover $(C^1_{n}+D^2_n)$ for any $n=0,1,\ldots,4$, using the MDS property of code $\mathbb{C}$. Similarly, $C^2_{n}+D^0_{n}$ can be recovered. Therefore, the following information on the requested message $A$ can be obtained
\begin{align}
\left[
\begin{array}{c|c|c|c|c}
* & *      & A^0_{2}     & A^1_{3}          & A^2_{4}\\
* & *      & A^1_{2}    & A^2_{3}           & A^0_{4}
\end{array}
\right],
\end{align}
from which message $A$ can clearly be reconstructed. 

\subsubsection{Correctness, Privacy, and Communication Costs}

The following lemma establishes the correctness of Construction-B when $T\leq N-T$, the proof of which can be found in Appendix \ref{appendix:forward}.
\begin{lemma}
\label{lemma:lowrate}
In the construction above, for any request of message-$k^*$ and any random key $\Fb$, 
    \begin{enumerate}
        \item $|\Tc_i| = T$ for any $i\in \{0,1,\ldots,s-1\}$;
        \item $|\Nc_m| = T$ for any $m \in \{0,1,\ldots, r-1\}$.
    \end{enumerate}
\end{lemma}

\begin{theorem} \label{thm:low2}
Construction-B is both private and capacity-achieving for $T\leq N-T$.
\end{theorem}
\begin{proof} 
The fact that the code is private is immediate for the same reason for the case $T\geq N-T$. 
The expected length of the answers is
\begin{align}
&\sum_{n=0}^{N-1} \mathbb{E} (\ell_n)=s \sum_{n=0}^{N-1} \mathbf{Pr}\left(\min_{k=0,1,\ldots,K-1} Q^{[k^*]}_{k,n}<r\right)\nonumber,
\end{align}
assuming an arbitrary message $k^*$ is being requested. 
By the definition of $Q^{[k^*]}_{n}$, if any item in
\begin{align*}
\ceil{(\Fb_0,...,\Fb_{k^*-1},\Fb_{k^*+1},\ldots,\Fb_{K-1})_{r+s}}_r
\end{align*}
is less than $r$, then $\min_{k=0,1,\ldots,K-1} {Q}_{k,n}^{[k^*]}<r$ for all $n=0,1,\ldots,N-1$, which will induce $sN$ transmitted symbols in the retrieval from all databases; this event $E$ occurs with probability $1-(s/(r+s))^{K-1}$. 

On the other hand, when the event $E$ does not occur, in the vector
\begin{align*}
\ceil{(\Fb_{k^*}+0,\Fb_{k^*}+1,\ldots,\Fb_{k^*}+N-1)_{r+s}}_r
\end{align*}
the number of elements that are less than $r$ is exactly $N-T$, which induces $s(N-T)$ symbols being transmitted. Therefore
\begin{align}
&\sum_{n=0}^{N-1} \mathbb{E} (\ell_n)=\mathbf{Pr}(E)sN+\left(1-\mathbf{Pr}(E)\right)s(N-T)\nonumber\\
&\qquad=sN-sT\left(\frac{T}{N}\right)^{K-1}=sN\left[1-\left(\frac{T}{N}\right)^{K}\right],
\end{align}
from which it follows that Construction-B is indeed capacity achieving.
\end{proof}

\begin{lemma}
   The upload cost of Construction-B for $T\leq N-T$ is upper-bounded by $$\min[N(K-1)\log(s+r),NK\log (r+1)].$$
\end{lemma}
The proof follows the same argument as that of Lemma \ref{lemma:uploadcost}, and it is omitted here for brevity.

\section{Minimum Message Size for Capacity-Achieving Linear Codes}

In this section, we establish the minimum message size as $\lcm(N-T,T)$ when $K$ is above a threshold, then shows that it is in fact possible to use an even smaller message size when $K$ is below this threshold. 

\subsection{Properties of Capacity-Achieving Linear MDS-PIR Codes} \label{sec:pro}

In this section, we provide two key properties of capacity-achieving linear MDS-PIR codes, which play an instrumental role in our study of the minimum message size.

\begin{lemma} \label{lem:simple}
Any linear MDS-PIR code must have:
\begin{enumerate}[label={\bfseries P\arabic*},start=0]
    \item For any $\Tc \subseteq \{0, 1, \ldots, N-1\}$ satisfying $|\Tc| = T$, $\{A_n^{[k]}\}_{n \in \Tc}$ are mutually independent, given any subset of messages $W^{0:K-1}$. \label{pro:mds}
\end{enumerate}
\end{lemma}

\begin{lemma}\label{lem:ite}
Let $\pi: \{0,1,\ldots, K-1\} \rightarrow \{0,1,\ldots, K-1\}$ be a permutation function. We have for any $k=0,1,\ldots, K-2$,
\begin{align}
&N\left[ \sum_{n=0}^{N-1} H(A_{n}^{[\pi(k)]} \mid W^{\pi(0:k-1)}, \Fb) - L \log|\Xc|\right]  \nonumber\\
&\qquad\qquad\qquad\geq T \sum_{n=0}^{N-1}H( A_{n}^{[\pi(k+1)]} \mid W^{\pi(0:k)}, \Fb). \label{eqn:ite}
\end{align}
Moreover, for any linear MDS-PIR code for which the equality holds for any $k$ and $\pi(\cdot)$ in (\ref{eqn:ite}), let $q_{0:N-1}$ be a combination of queries such that $\mathbf{Pr}(q_{0:N-1})>0$ for the retrieval of $W^{k^*}$, 
then the code must have: 
\begin{enumerate}[label={\bfseries P\arabic*}]
    \item For any $\Tc \subseteq \{0,1,\ldots,N-1\}$ such that $|\Tc| = T$, and $\Jc \subseteq \{0,1,\ldots, K-1\}$ satisfying $k^* \in \Jc$ \label{pro:3}
\begin{align*}
    H\left( A^{(q_{n'})}_{n'} ,n' \in \bar\Tc \mid W^{\Jc}, A_n^{(q_n)}, n \in \Tc\right)=0,
\end{align*}
where $\bar{\Tc}$ is the complement of $\Tc$.
\end{enumerate}
\end{lemma}

Property {\it \ref{pro:mds}} is a direct consequence of the linear MDS-PIR code definition. Property {\it \ref{pro:3}} states that the interference signals from the answers of any $T$ databases in a capacity-achieving code can fully determine all interference signals in other answers. The inequalities in Lemma \ref{lem:ite} are the key steps in deriving the capacity of MDS-PIR codes; the proofs of these properties are can be found in Appendix \ref{appendix:converse}. Conversely, for any capacity-achieving linear MDS-PIR code, these inequalities must hold with equality, implying the following theorem.

\vspace{0.1cm}
\begin{theorem}\label{thm:cap}
Any capacity-achieving  linear MDS-PIR code must have the properties {\it \ref{pro:mds}} and {\it \ref{pro:3}}.
\end{theorem}
\begin{proof}
Let $\pi:\{0,1,...,K-1\} \rightarrow \{0,1,...,K-1\}$ be a permutation. By applying Lemma \ref{lem:ite} recursively, we can write
\begin{align}
&\frac{L \log|\Xc|}{R} \geq \sum_{n=0}^{N-1} H(A_{n}^{[\pi(0)]} \mid \Fb) \\
&\geq L\log|\Xc| + \frac{T}{N} \sum_{n=0}^{N-1} H(A_{n}^{[\pi(1)]}\mid W^{\pi(0)}, \Fb)\\
&\geq \cdots\\
&\geq L\log|\Xc|\left(1+ \frac{T}{N} + \cdots+ \left(\frac{T}{N}\right)^{K-1} \right),
\end{align}
and it follows that $R\geq C$. For any capacity-achieving linear MDS-PIR code, all the inequalities in Lemma~\ref{lem:ite} must be equality. Therefore, any capacity-achieving linear MDS-PIR code must have properties {\it \ref{pro:mds}} and {\it \ref{pro:3}}.
\end{proof}

A similar set of properties for capacity-achieving PIR codes on replicated databases was derived in \cite{tian2018capacity}, which holds for a more general code class and in a more explicit form. For MDS coded databases, firstly it is more meaningful to consider only linear codes, and secondly, it is not clear whether the properties stated in Lemma \ref{lem:simple} and \ref{lem:ite} can be extended to the more general class of codes considered in \cite{tian2018capacity}.

\subsection{Bounding the Minimum Message Size}
\label{sec:lowerbound}

The main result of this section is the following theorem. 

\vspace{0.1cm}
\begin{theorem}\label{thm:OuterForMessageSize}
When $K>T/\gcd1(N,T)$, the message size of any capacity-achieving linear MDS-PIR code satisfies $L\geq \lcm(N-T,T)$.
\end{theorem}

From Theorem \ref{thm:OuterForMessageSize}, we can conclude that codes obtained by Construction-A and Construction-B indeed have the minimum message size when $K > T/\gcd1(N,T)$. The proof of Theorem \ref{thm:OuterForMessageSize} relies on the delicate relation among a set of auxiliary quantities $H_n^k$'s and $I_n^k$'s which we define next. For any given capacity-achieving linear MDS-PIR code, let $(\breve{k},\breve{f})$ be the maximizer for the following optimization problem for $n=0$:
\begin{align}
\max_{k=0,1,\ldots,K-1} \max_{f \in \Fc} H(A_n^{[k]} \mid W^{k}, \Fb=f).\label{eqn:max}
\end{align}
Define for $k=0,1,\ldots, K-1$ and $n=0,1,\ldots,N-1$,
\begin{align*}
H_n^k := \frac{H(A_n^{[\breve{k}]} \mid W^k, \Fb=\breve{f})}{\log|\Xc|},\, I_n^{k} := \frac{I(A_n^{[\breve{k}]}; W^k \mid \Fb=\breve{f})}{\log|\Xc|}.
\end{align*}

The following lemma implies that the optimization problem in (\ref{eqn:max}) has the same maximizer for all $n\in\{0,1,\ldots,N-1\}$. 

\vspace{0.1cm}
\begin{lemma}\label{lem:eq_ent}
For any capacity-achieving linear MDS-PIR code, $\forall n' \not= n''$ where $n',n'' \in \{0,1,\ldots, N-1\}$, any $k^*\in \{0,1,\ldots,K-1\}$, any $f\in \mathcal{F}$, 
\begin{eqnarray*}
H(A_{n'}^{[k^*]} \mid W^{k^*}, \Fb=f) = H(A_{n''}^{[k^*]} \mid W^{k^*}, \Fb=f).
\end{eqnarray*}
\end{lemma}

This lemma also implies that we can define $H^{\breve{k}}:=H_0^{\breve{k}}=\ldots=H_{N-1}^{\breve{k}}$. The next two lemmas establish a critical property of, and relevant relations between, $H_n^k$'s and $I^k_n$'s. 

\vspace{0.1cm}
\begin{lemma}
\label{lemma:keyeq}
\begin{align}
L-(N-T)H^{\breve{k}}=\sum_{n=0}^{N-1}I_n^{\breve{k}}. \label{eqn:keyeq}
\end{align}
\end{lemma}

\vspace{0.1cm}
\begin{lemma} \label{lem:Dom}
For any $k\in \{0,1,\ldots,K-1\}$ and $n\in \{0,1,\ldots,N-1\}$, $H_n^k$ and $I_n^k$ are integers; moreover
\begin{align}
    & H_{n}^{\breve{k}} \geq \sum_{k \not= \breve{k}} I_n^k, \label{eqn:firstineq}
\end{align}
and when $K>s$, 
\begin{align}
H^{\breve{k}} \geq s\label{eqn:secondineq}
\end{align}
\end{lemma}

The proofs of Lemma \ref{lem:eq_ent}-\ref{lem:Dom} are given in Appendix \ref{appendix:converse}. We are now ready to prove Theorem \ref{thm:OuterForMessageSize}. 
\begin{proof}[Proof of Theorem \ref{thm:OuterForMessageSize}]

When $K>s$,  by (\ref{eqn:secondineq}) in Lemma \ref{lem:Dom} and Lemma \ref{lemma:keyeq}, we have
\begin{eqnarray}
L - (N-T)s \geq L - (N-T)H^{\breve{k}}=\sum_{n=0}^{N-1}I_n^{\breve{k}}\geq 0.
\end{eqnarray}
Substituting $L=Mps$ and $(N-T)=pr$ into the left hand side leads to the conclusion $M\geq r$, implying that $L\geq MT\geq rT=\lcm(N-T,T)$.
\end{proof}

\subsection{Message Size Reduction for Small $K$}
\label{sec:small}

The following theorem confirms that for small $K$, it is in fact possible to construct a capacity-achieving code with an even smaller message size. 

\begin{theorem}\label{thm:4}
When $K= 2$ and $T \geq N-T$, the minimum message size of capacity-achieving codes is $T$.
\end{theorem}

\begin{proof}[Proof of Theorem \ref{thm:4}]
Since $L=MT>0$, it is trivial to see that $L \geq T$, and thus it only remains to provide a construction of a capacity-achieving linear MDS-PIR code with such a message size.

Let database-$n$ store two symbols $V_n^0, V^1_n \in \Xc$, which are MDS-coded symbols of messages $W^0$ and $W^1$, respectively. When the user wishes to retrieve message $W^{k^*}$ where $k^* \in \{0,1\}$, two query strategies are used. 
\begin{itemize}
\item With probability $\frac{T}{N}$, randomly partition $N$ databases into $3$ disjoint sets $\Gc^{(0)}$, $\Gc^{(1)}$ and $\Gc^{(2)}$, with $|\Gc^{(0)}|=N-T$, $|\Gc^{(1)}|=2T-N$ and $|\Gc^{(2)}|=N-T$. The user requests $V^0_n \oplus V^1_n$ from databases in $\Gc^{(0)}$, $(V^0_n,V^1_n)$ from those in $\Gc^{(1)}$, and $V^{1-k^*}_n$ in $\Gc^{(2)}$. 
\item With probability $\frac{N-T}{N}$, randomly partition $N$ databases into $2$ disjoint sets $\Gc^{(3)}$ and $\Gc^{(4)}$, with $|\Gc^{(3)}|=T$ and  $|\Gc^{(4)}|=N-T$. The user requests $V^{k^*}_n$ from databases in $\Gc^{(3)}$, but no information from those in $\Gc^{(4)}$. 
\end{itemize}
It is straightforward to verify that the code is indeed correct, private, and capacity-achieving. 
\end{proof}

Theorem \ref{thm:4} provides a code construction with a small message size for the special case of $K=2$ and $T\geq N-T$, however, we suspect codes with small message sizes also exist for other parameters when $K$ is below the threshold,  but they may require certain more sophisticated combinatorial structure. A systematic approach to construct such codes and a converse result appear rather difficult to find.

\section{Conclusion}
\label{sec:con}

We proposed two code constructions for private information retrieval from MDS-coded databases with message size $\lcm(N-T,T)$, and show that this is the minimum message size for linear codes when $K$ is above a threshold. For smaller $K$ it is in fact possible to design PIR codes with an even smaller message size, which we show by a special example for $K=2$. This work generalizes our previous result on private information retrieval from replicated databases \cite{tian2018capacity} in a highly non-trivial manner. We expect the code constructions and the converse proof approach to be also applicable and fruitful in other privacy-preserving primitives. 

Independent of this work and inspired by our previous result \cite{tian2018capacity}, Zhu et al.  \cite{Zhu:MDSPIR19} discovered a different code construction, and also derived a lower bound on the message size similar to the one reported here. All three code constructions have the same message size, however, the two constructions we provided have a lower upload cost due to the queries being better compressed. It is worth noting that in our previous work \cite{tian2018capacity}, the proposed code was also shown to be optimal in terms of the upload cost, however, proving the proposed codes in the current work to be optimal appears more difficult due to the more complex dependence stipulated by the MDS code.

\appendix

\section{Proofs of Lemma \ref{lemma:TN}, Lemma \ref{lemma:Double}, and Lemma \ref{lemma:lowrate}}
\label{appendix:forward}
 \begin{proof}[Proof of Lemma \ref{lemma:TN}]
    Fix a particular $k^*$.  It is convenient to represent the $\tilde{Q}^{k^*,i}_{n}$ as a matrix given below
    \begin{align}
\begin{bmatrix}
    \tilde{Q}^{k^*,0}_{0}      & \tilde{Q}^{k^*,0}_{1} & \ldots & \tilde{Q}^{k^*,0}_{N-1} \\
    \tilde{Q}^{k^*,1}_{0}      & \tilde{Q}^{k^*,1}_{1} & \ldots & \tilde{Q}^{k^*,1}_{N-1} \\
     \vdots                              & \vdots                          & \vdots & \vdots \\
    \tilde{Q}^{k^*,s-1}_{0}      & \tilde{Q}^{k^*,s-1}_{1} & \ldots & \tilde{Q}^{k^*,s-1}_{N-1} 
\end{bmatrix}. \label{eqn:Qki}
    \end{align}
    Since $\tilde{Q}^{k^*,i}_{n}=(\Fb_{k^*} + i + n)_{s+r}$, $N = p\cdot(s+r)$, and $T=p\cdot s$,  for any $i$ and any given realization of key $\Fb$, there are exactly $T$ elements in $\tilde{Q}^{k^*,i}_{0:N-1}$, which is a row of (\ref{eqn:Qki}), that are greater than or equal to $r$.  This proves the first statement that $|\mathcal{T}_i|\geq T$. For the second statement, consider a fixed  $m\in \{0,1,\ldots,r+s-1\}$. It is clear that each row in the matrix (\ref{eqn:Qki}) has exactly $p$ positions being $m$. Therefore, there are a total of $p\cdot s=T$ elements in the matrix being $m$. Since for any $i\neq i'$ where $0\leq i,i'< s$, we have $\tilde{Q}^{k^*,i}_{n}\neq \tilde{Q}^{k^*,i'}_{n}$, due to $\tilde{Q}^{k^*,i}_{n}=(\Fb_{k^*} + i + n)_{s+r}$, these $T$ positions are all in different columns, implying $|\Nc_m| = T$.  
    \end{proof}

\begin{proof}[Proof of Lemma \ref{lemma:Double}]
For each fixed $i\in \{0,1,2,\ldots,s-1\}$, we have
\begin{align}
\Tcc_i=\left\{n\big{|}~P_{i, Q_{k^*,n}^{[k^*]}} = 0\right\}=\left\{n\big{|}~\bar{P}_{i, \bar{Q}_{k^*,n}^{[k^*]}} = 0\right\},
\end{align}
however, the vector
\begin{align}
(\bar{Q}_{k^*,0}^{[k^*]}, \bar{Q}_{k^*,1}^{[k^*]},\ldots,\bar{Q}_{k^*,n-1}^{[k^*]})\label{eqn:tildeQkstar}
\end{align}
is simply a permutation (in fact, a cyclic shift) of $p$-replicated vector of $(0,1,\ldots,r+s-1)$, and thus 
\begin{align}
\left(\bar{P}_{i, \bar{Q}_{k^*,0}^{[k^*]}},\bar{P}_{i, \bar{Q}_{k^*,1}^{[k^*]}},\ldots,\bar{P}_{i, \bar{Q}_{k^*,N-1}^{[k^*]}}\right) \label{eqn:tzeros}
\end{align}
is in fact is a permutation of $[\bar{P}_i,\bar{P}_i,\ldots,\bar{P}_i]$, i.e., a $p$-replicated version of the $i$-th row of $\bar{P}$. It is now clear that $|\Tcc_i|=T$, because each row of $\bar{P}$ has exactly $s$ zeros, and the replication gives a total of $ps=T$ zeros in the vector in (\ref{eqn:tzeros}).

To see that $|\Nc| = T$, let us focus on any single database-$n$ such that $|\mathcal{S}_n|>0$, i.e., $\bar{P}_{i, \bar{Q}_{k^*,n}^{[k^*]}}=1$ for some $i\in\{0,1,\ldots,s-1\}$. This condition is equivalent to $\bar{Q}_{k^*,n}^{[k^*]}<s$, and moreover, it also implies $|\mathcal{S}_n|=r$ because the vector
\begin{align}
\left(\bar{P}_{0, \bar{Q}_{k^*,n}^{[k^*]}},\bar{P}_{1, \bar{Q}_{k^*,n}^{[k^*]}},\ldots,\bar{P}_{s-1, \bar{Q}_{k^*,n}^{[k^*]}}\right)^T,
\end{align}
is the $j=\bar{Q}_{k^*,n}^{[k^*]}$-th column of the extended pattern matrix $\bar{P}$, and for any $j<s$, such a column has exact $r$ positions being $1$. It is also immediately clear that among the $N$ databases, there are a total of $sp=T$ of them with $\bar{Q}_{k^*,n}^{[k^*]}<s$, again because the vector (\ref{eqn:tildeQkstar})
is a permutation of the $p$-replicated vector of $(0,1,\ldots,r+s-1)$.
\end{proof}

\begin{proof}[Proof of Lemma \ref{lemma:lowrate}]
Define $\Tc := \left\{n| Q^{[k^*]}_{k^*,n}=r\right\}$, and notice that in Construction-B, $\Tc_i = \Tc$ for all $i=0,1,\ldots,s-1$. Moreover
\begin{align}
\Tc = \left\{n\big{|} Q^{[k^*]}_{k^*,n}=r\right\}=\left\{n\big{|} \bar{Q}^{[k^*]}_{k^*,n}\geq r\right\}.
\end{align}
Notice that the vector 
\begin{align}
\left(\bar{Q}^{[k^*]}_{k^*,0},\bar{Q}^{[k^*]}_{k^*,1},\ldots,\bar{Q}^{[k^*]}_{k^*,N-1}\right)
\end{align}
is a permutation of $p$-replicated vector $(0,1,\ldots,r+s-1)$, and thus it has exactly $ps=T$ items that are greater or equal to $r$. This directly implies $|\Tc|=T$.

By definition, for any $ \tilde{Q}_{n}^{k^*,i}=m$ where $m<r$, to hold, we must have  $(\Fb_{k^*}+n)_{s+r}<r$. It is also clear that $\tilde{Q}^{k^*,i}_{n} = ((\Fb_{k^*}+n)_{s+r}+i)_r$. For any $m\in \{0,1,\ldots,r-1\}$, there are in fact exactly $ps=T$ pairs of $(n,i)$'s such that $\tilde{Q}^{k^*,i}_{n}=m$. This can be seen as follows: the vector 
\begin{align}
(\Fb_{k^*}+0,\Fb_{k^*}+1,\ldots,\Fb_{k^*}+N-1)_{s+r}
\end{align}
has $rp=(N-T)$ items less than $r$, and these items are a permutation of $p$-replicated vector $(0,1,\ldots,r-1)$. The symmetry of these values indeed implies that there are $(N-T)s/r=T$ such $(n,i)$ pairs for each $m\in \{0,1,\ldots,r-1\}$. We next show that such pairs do not have any common $n$. To see this, suppose there are two distinct pairs $(n,i)$ and $(n,i')$, which satisfy 
\begin{align}
((\Fb_{k^*}+n)_{s+r}+i)_r=((\Fb_{k^*}+n)_{s+r}+i')_r=m,
\end{align}
for certain $m\in \{0,1,\ldots,r-1\}$. 
However, notice that 
$0\leq i,i'<s\leq r$, we therefore must have $i=i'$, which contradicts our supposition. It follows that indeed $|\Nc_m|=T$ for any $m \in \{0,1,\ldots, r-1\}$. 
\end{proof}

\section{Proofs of Lemmas \ref{lem:ite}-\ref{lem:Dom}}
\label{appendix:converse}
\begin{proof}[Proof of Lemma \ref{lem:ite}]
\begin{align}
&N\left[ \sum_{n=0}^{N-1} H(A_{n}^{[\pi(k)]} \mid W^{\pi(0:k-1)}, \Fb) - L\log|\Xc|\right] \label{eqn:conv0}\\
&\geq N\left[ H(A_{0:N-1}^{[\pi(k)]} \mid W^{\pi(0:k-1)}, \Fb) - L\log|\Xc|\right] \label{eqn:conv1}\\
&= N {\Big [} H(A_{0:N-1}^{[\pi(k)]} \mid W^{\pi(0:k-1)}, \Fb) \notag \\
&\quad\quad - I (W^{\pi(k)} ; A_{0:N-1}^{[\pi(k)]} \mid W^{\pi(0:k-1)}, \Fb) {\Big ]} \\
&= N H(A_{0:N-1}^{[\pi(k)]} \mid W^{\pi(0:k)}, \Fb) \\
&\overset{(b)}{\geq} \sum_{n=0}^{N-1} H(A_{\rho((n:n+T-1)_N)}^{[\pi(k)]} \mid W^{\pi(0:k)}, \Fb) \label{eqn:conv_rho} \\
&\overset{(c)}{=}  \sum_{n=0}^{N-1} \sum_{s=0}^{T-1} H(A^{[\pi(k)]}_{(n+s)_N} | W^{\pi(0:k)}, \Fb) \label{eqn:conv2} \\
&= T \left[ \sum_{n=0}^{N-1} H(A_{n}^{[\pi(k)]} \mid W^{\pi(0:k)}, \Fb)\right]  \\
&= T \left[ \sum_{n=0}^{N-1} H(A_{n}^{[\pi(k+1)]} \mid W^{\pi(0:k)}, \Fb)\right] \label{eqn:conv3},
\end{align}
where $\rho: \{0,1,\ldots,N-1\} \rightarrow \{0,1,\ldots,N-1\}$ in inequality (\ref{eqn:conv_rho}) is a permutation over $\{0,1,\ldots,N-1\}$.
Equality (\ref{eqn:conv3}) is due to the privacy constraint, which is
\begin{eqnarray}
&&H(A_n^{[\pi(k)]} \mid W^{\pi(0:k)}, \Fb)\\
&=& \sum_{f \in \Fc} \mathbf{Pr}(\Fb=f)H(A_n^{[\pi(k)]}\mid W^{\pi(0:k)}, \Fb=f)\\
&=&\sum_{q_n \in \Qc_n} \mathbf{Pr}(Q_n=q_n)H(A_n^{(q_n)} \mid W^{\pi(0:k)})\\
&=&H(A^{[\pi(k+1)]}_n \mid W^{\pi(0:k)}, \Fb)
\end{eqnarray}

For the equality $(b)$ to hold for an MDS-PIR code, i.e.,
\begin{align*}
H(A^{[\pi(k)]}_{0:N-1}\mid W^{\pi(0:k)}, \Fb) = H(A^{[\pi(k)]}_{\rho((n:n+T-1))_N} \mid W^{\pi(0:k)}, \Fb).
\end{align*}
the equality must hold for each $\Fb=f$, which concludes the proof of property~\ref{pro:3}.
\end{proof}

\begin{proof}[Proof of Lemma \ref{lem:eq_ent}]
Let $\Tc',\Tc'' \subseteq \{0,1,\ldots,N-1\}$ such that $|\Tc'|=|\Tc''| = T$,  and $n'\in \Tc'$, $\Tc'' = \Tc'\setminus \{n'\} \cup \{n''\}$. 
By \ref{pro:3}, any capacity-achieving linear code must have 
\begin{align}
&H(A^{[k^*]}_{\Tc'} \mid A^{[k^*]}_{\Tc''}, W^{k^*}, \Fb = f)\nonumber\\
&=H(A^{[k^*]}_{\Tc''} \mid A^{[k^*]}_{\Tc'}, W^{k^*},\Fb = f) =0,
\end{align}
which implies 
\begin{eqnarray}
H(A^{[k^*]}_{\Tc'} \mid W^{k^*},\Fb = f)=  H(A^{[k^*]}_{\Tc''} \mid W^{k^*},\Fb = f).
\end{eqnarray}
Invoking \ref{pro:mds} leads to
\begin{align*}
\sum_{n \in \Tc'} H(A_n^{[k^*]} \mid W^{k^*}, \Fb = f)=\sum_{n \in \Tc''} H(A_n^{[k^*]} \mid W^{k^*}, \Fb = f),
\end{align*}
which further implies that 
\begin{align*}
H(A_{n'}^{[k^*]} \mid W^{k^*}, \Fb=f) = H(A_{n''}^{[k^*]} \mid W^{k^*}, \Fb=f).
\end{align*}
This completes the proof. 
\end{proof}

\begin{proof}[Proof of Lemma \ref{lemma:keyeq}]
For any capacity-achieving code, (\ref{eqn:conv0}) must equal to (\ref{eqn:conv2}). The equality should also hold when $\Fb=\breve{f}$, $k=0$ and $\pi(0)=\breve{k}$. Substituting the definition of $H_n^{\breve{k}}$ and $I_n^{\breve{k}}$ directly, we have
\begin{align*}
    N \left[\sum_{n=0}^{N-1}\left(H^{\breve{k}} + I_n^{\breve{k}}\right) -L \right] &= \sum_{n=0}^{N-1}\sum_{s=0}^{T-1} H^{\breve{k}},
\end{align*}
which simplifies to the desired equality.
\end{proof}

\begin{proof}[Proof of Lemma \ref{lem:Dom}]
Let $q_n = \phi_n(\breve{k},\breve{f})$. The linearity of the code implies that
\begin{align}
H_n^k= \frac{H(A_n^{[\breve{k}]} \mid W^k, \Fb=\breve{f})}{\log|\Xc|}=\frac{H(A_{n}^{(q_n)}|W^k)}{\log|\Xc|}
\end{align}
is an integer. Notice that $I_n^k+H_n^k$ is also an integer by the same argument, from which we conclude that $I_n^k$ is also an integer. 

To see the inequality in (\ref{eqn:firstineq}), we write
\begin{align}
&H(A_n^{(q_n)})\geq I(A_n^{(q_n)}; W^{0:K-1}) \notag \\
&\qquad=\sum_{k=0}^{K-1} I(A_{n}^{(q_n)} ; W^k|W^{0:k-1})\notag\\
&\qquad\stackrel{(a)}{\geq} \sum_{k=0}^{K-1} I(A_{n}^{(q_n)} ; W^k)= \log|\Xc| \sum_{k=0}^{K-1} I_{n}^{k},
\end{align}
where $(a)$ is because
\begin{align}
&I(A_{n}^{(q_n)} ; W^k|W^{0:k-1})\notag\\
&=H( W^k|W^{0:k-1})-H( W^k|W^{0:k-1},A_{n}^{(q_n)} )\notag\\
&\geq H( W^k)-H( W^k|A_{n}^{(q_n)} )\notag\\
&=I(A_{n}^{(q_n)} ; W^k).
\end{align}
It follows that
\begin{eqnarray}
H^{\breve{k}}_n \log|\Xc|&=&H(A_n^{(q_n)}) - I_{n}^{\breve{k}}\log|\Xc|\notag\\
&\geq& \log|\Xc| \sum_{k\not = \breve{k}}I_{n}^{k}.
\end{eqnarray}

Next we prove the inequality (\ref{eqn:secondineq}).  First define the following query support set at database-$n$
\begin{eqnarray}
\Qc^*_n = \{q_n| q_n=\phi_n(k=0, f) \mbox{ for some } f\in \Fc \},
\end{eqnarray}
and since the code is privacy preserving, it is also the query support set for all other $k=1,\ldots,K-1$. We can then write
\begin{align}
&H^{\breve{k}}\log|\Xc|=\max_{k=0,1,\ldots,K-1} \max_{f \in \Fc} H(A_n^{[k]} \mid W^{k}, \Fb=f)\notag\\
&=\max_{k=0,1,\ldots,K-1} \max_{q_n \in \Qc^*_n} H(A_n^{(q_n)} \mid W^{k}) \notag\\
&\geq\max_{q_n \in \Qc^*_n} H(A_n^{(q_n)} \mid W^k)\geq H(A_n^{[\breve{k}]} \mid W^k, \Fb=\breve{f})\notag\\
&=H_{n}^{k}\log|\Xc| \label{eqn:dom}
\end{align}
for any $n=0,1,\ldots,N-1$ and $k=0,1,\ldots,K-1$.

Plugging $T=ps$, $N-T=pr$, and $L=MT$ into (\ref{eqn:keyeq}) gives
\begin{align}
\sum_{n=0}^{N-1} I_n^{\breve{k}}=L-(N-T)H^{\breve{k}}= p [Ms-rH^{\breve{k}}].
\end{align}
Since $I_{n}^{\breve{k}} \geq 0$, the left hand side is strictly positive unless $(N-T)$ is a factor of $L$. If the left hand side is zero, then $L\geq \lcm(T,N-T)$, implying $H^{\breve{k}}\geq s$. 
If $(N-T)$ is not a factor of $L$, implying that the left hand side is indeed strictly positive, 
then there exists at least one $n^* \in \{0,1,\ldots,N-1\}$ such that $I_{n^*}^{\breve{k}} \geq 1$. 
For any $k=0,1,\ldots,K-1$,
\begin{eqnarray}
I_{n^*}^k + H_{n^*}^k=H(A_{n^*}^{[\breve{k}]} \mid \Fb=\breve{f})/\log|\Xc|,\label{eqn:sumconstant}
\end{eqnarray}
however, the right hand side of (\ref{eqn:sumconstant}) is independent of $k$, and thus $I_{n^*}^k + H_{n^*}^k$ is in fact a constant. Furthermore, $H^{\breve{k}}\geq H_{n^*}^k$ by (\ref{eqn:dom}), it follows that
$I_{n^*}^k \geq I_{n^*}^{\breve{k}}$ for all $k=0,1,\ldots,K-1$. By (\ref{eqn:firstineq}), we can write
\begin{align}
H_{n^*}^{\breve{k}}\geq \sum_{k \not= k^*} I_{n^*}^k \geq (K-1) I_{n^*}^{\breve{k}}\geq (K-1)\geq s,
\end{align}
when $K>s$. 
\end{proof}

\bibliographystyle{IEEEtran}
 \newcommand{\noop}[1]{}

\end{document}